\def\IEEE{1} 

\if 1\IEEE
\documentclass[letterpaper, 10 pt, conference]{ieeeconf}  
\else
\documentclass[10pt]{article}
\fi

\usepackage[utf8]{inputenc}
\usepackage{amsfonts}
\usepackage[english]{babel}
\usepackage{amsmath}
\usepackage{amssymb}

\if 0\IEEE
\usepackage{amsthm}
\usepackage{verbatim}
\fi

\usepackage{tikz}
\usepackage{mathtools}
\DeclareMathAlphabet{\pazocal}{OMS}{zplm}{m}{n}
\usepackage [autostyle, english = american]{csquotes} 
\MakeOuterQuote{"}
\usepackage{xargs}

\if 0\IEEE
\usepackage[numbers]{natbib}
\usepackage{enumitem}
\usepackage[backref=page]{hyperref}
\hypersetup{					 
   colorlinks,
   linkcolor={red},
   citecolor={blue},
   urlcolor={blue!80!black}
}
\fi

\if 1\IEEE
\IEEEoverridecommandlockouts 
\usepackage{graphics} 
\usepackage{epsfig} 
\usepackage{mathptmx} 
\usepackage{times} 

\usepackage{amsthm}
\usepackage{cleveref}
\usepackage{resizegather}
\newtheorem{theorem}{Theorem}[section]
\newtheorem{corollary}[theorem]{Corollary}
\newtheorem{lemma}[theorem]{Lemma}
\newtheorem{proposition}[theorem]{Proposition}
\theoremstyle{definition}
\newtheorem{definition}{Definition}[section]
\theoremstyle{definition}
\newtheorem{example}{Example}[section]
\theoremstyle{remark}
\newtheorem{remark}{Remark}

\fi

\if 0\IEEE
\newtheorem{theore}{Theorem}[section]
\newenvironment{theorem}
{%
\pushQED{\qed}\begin{theore}}
{\popQED\end{theore}}  

\newtheorem{cor}[theore]{Corollary}

\newtheorem{lem}[theore]{Lemma}

\newtheorem{pro}[theore]{Proposition}
\newenvironment{proposition}
{%
\pushQED{\qed}\begin{pro}}
{\popQED\end{pro}}

\theoremstyle{definition}
\newtheorem{defi}{Definition}[section]
\newenvironment{definition}
{%
\pushQED{\qed}\begin{defi}}
{\popQED\end{defi}}

\newtheorem{exa}{Example}[section]

\theoremstyle{remark}
\newtheorem{rem}{Remark}
\newenvironment{remark}
{%
\pushQED{\qed}\begin{rem}}
{\popQED\end{rem}}
\fi

\usetikzlibrary{positioning,arrows,petri,calc,decorations.markings,arrows.meta}
\tikzset{
place/.style={
circle,
thick,
minimum size=6mm,
draw
},
transitionV/.style={
rectangle,
thick,
fill=black,
minimum height=6mm,
inner xsep=1pt
}
}

\if 0\IEEE
\renewcommand*{\backref}[1]{}
\renewcommand*{\backrefalt}[4]{%
    \ifcase #1 (Not cited.)%
    \or        (Cited on page~#2.)%
    \else      (Cited on pages~#2.)%
    \fi}
\fi

\title{
\if 1\IEEE
\LARGE \bf%
\fi
	Periodic Trajectories in P-Time Event Graphs\\and the Non-Positive Circuit Weight Problem} 
\author{Davide Zorzenon, Jan Komenda and J\"{o}rg Raisch%
\thanks{This work was supported in part by the Deutsche Forschungsgemeinschaft (DFG, German Research Foundation), Projektnummer RA 516/14-1; in part by the Grantov\'{a} agentura \v{C}esk\'{e} republiky (GACR, Czech Science Foundation) grant 19-06175J; in part by MSMT INTER-EXCELLENCE project LTAUSA19098; and in part by RVO 67985840.}
\thanks{D. Zorzenon and J. Raisch are with Technische Universit\"at Berlin,
        Control Systems Group, Einsteinufer 17, D-10587 Berlin, Germany (e-mail: zorzenon@control.tu-berlin.de; raisch@control.tu-berlin.de).}
\thanks{J. Komenda is with Institute of Mathematics,
        Czech Academy of Sciences, \v{Z}itn\'a 25, 115 67 Prague, Czech Republic (e-mail: komenda@ipm.cz).}
}

\begin{document}

\newcommand{\graph}{\pazocal{G}}
\newcommand*{\Resize}[2]{\resizebox{#1}{!}{$#2$}}%

\newcommand{\R}{\mathbb{R}}
\newcommand{\nat}{\mathbb{N}}
\newcommand{\nato}{\mathbb{N}_0}
\newcommand{\Rmax}{{\R}_{\normalfont\fontsize{7pt}{11pt}\selectfont\mbox{max}}}
\newcommand{\maxplus}{\oplus}
\newcommand{\maxtimes}{\otimes}
\newcommand{\minplus}{\wedge}
\newcommand{\mintimes}{\odot}
\newcommand{\bigmaxplus}{\bigoplus}
\newcommand{\bigmaxtimes}{\bigotimes}
\newcommand{\bigminplus}{\bigwedge}
\newcommand{\bigmintimes}{\bigodot}
\newcommand{\Rbar}{\bar{\R}_{\normalfont\fontsize{7pt}{11pt}\selectfont\mbox{max}}}
\newcommand{\emptystr}{\mathsf{e}}
\newcommand{\arcs}{E}
\newcommand{\nodes}{N}
\newcommand{\nonegset}{\Gamma}
\newcommand{\nonegsetm}{\Gamma_{M}}
\newcommand{\xy}{\pazocal{S}}
\newcommand{\xx}{\pazocal{P}}
\newcommand{\yy}{\pazocal{I}}
\newcommand{\ltm}{\mu}
\newcommand{\lang}{\pazocal{L}}
\newcommand{\Tpazo}{\pazocal{T}}
\newcommand{\MP}{P}
\newcommand{\MI}{I}
\newcommand{\MC}{C}
\newcommand{\calP}{\pazocal{\MP}}
\newcommand{\calI}{\pazocal{\MI}}
\newcommand{\calC}{\pazocal{\MC}}
\newcommand{\setP}{\mathcal{\MP}_\sigma}
\newcommand{\setI}{\mathcal{\MI}_\sigma}
\newcommand{\setC}{\mathcal{\MC}_\sigma}
\newcommand{\MA}{A}
\newcommand{\MB}{B}
\newcommand{\graphPIC}{\graph(\lambda \MP \oplus \lambda^{-1} \MI \oplus \MC)}
\newcommand{\graphcalPIC}{\graph(\lambda^d \calP \oplus \lambda^{-d} \calI \oplus \calC)}
\newcommand{\graphPICm}{\graph(\lambda \MP, \lambda^{-1} \MI, \MC)}
\newcommand{\graphNCPt}{\graph((\TP\totimes\MP)\oplus(\TI\totimes\MI)\oplus(\TC\totimes\MC))}
\newcommand{\solNCPt}{\Lambda_{\mbox{\normalfont\footnotesize NCP}}((\TP\totimes\MP)\oplus(\TI\totimes\MI)\oplus(\TC\totimes\MC))}
\newcommandx{\solNCP}[1]{\Lambda_{\mbox{\normalfont\footnotesize NCP}}(#1)}
\newcommandx{\solPIC}[1]{\Lambda_{\mbox{\normalfont\footnotesize PIC-NCP}}(#1)}
\newcommandx{\solPTEG}[1][1=d]{\Lambda_{\mbox{\normalfont\footnotesize P-TEG}_{#1}}(A^0,A^1,B^0,B^1)}
\newcommand{\wA}{\mathsf{\MakeLowercase{\MA}}}
\newcommand{\wB}{\mathsf{\MakeLowercase{\MB}}}
\newcommand{\wP}{\mathsf{\MakeLowercase{\MP}}}
\newcommand{\wI}{\mathsf{\MakeLowercase{\MI}}}
\newcommand{\wC}{\mathsf{\MakeLowercase{\MC}}}
\newcommand{\wZ}{\mathsf{z}}
\newcommand{\D}{\pazocal{D}}
\newcommand{\floor}[1]{\left\lfloor#1\right\rfloor}
\catcode`@=11
\def\plslash{\ifx\@currsize\normalsize
{\mathchoice
{\,\mbox{\raisebox{0.2ex}{$\scriptstyle\circ$}\kern-1ex$\setminus$}}
{\,\mbox{\raisebox{0.2ex}{$\scriptstyle\circ$}\kern-1ex$\setminus$}}%
{\,\mbox{\raisebox{0.14ex}{$\scriptscriptstyle\circ$}\kern-0.8ex%
${\scriptstyle\setminus}$}}%
{\,\mbox{\raisebox{0.14ex}{$\scriptscriptstyle\circ$}\kern-0.8ex%
${\scriptstyle\setminus}$}}}%
\else\ifx\@currsize\large\,\mbox{\raisebox{0.2ex}{$\scriptstyle\circ$}\kern-1ex$\setminus$}
\else\ifx\@currsize\small\,\mbox{\raisebox{0.2ex}{$\scriptstyle\circ$}\kern-1ex$\setminus$}
\else\,\mbox{\raisebox{0.2ex}{$\scriptstyle\circ$}\kern-0.1ex$\setminus$}
\fi\fi\fi}
\newcommand{\totimes}{\otimes^t}
\newcommandx{\TP}[1][1=d]{T_{#1}^{\MP}}
\newcommandx{\TI}[1][1=d]{T_{#1}^{\MI}}
\newcommandx{\TC}[1][1=d]{T_{#1}^{\MC}}
\newcommandx{\T}[2][2=d]{T_{#2}^{X_{#1}}}

\newcommand{\oldiv}{\protect\plslash}
\def\prslash{\ifx\@currsize\normalsize
{\mathchoice
{\mbox{\raisebox{0.2ex}{$\scriptstyle\circ$}\kern-1ex$/$}}
{\mbox{\raisebox{0.2ex}{$\scriptstyle\circ$}\kern-1ex$/$}}%
{\mbox{\raisebox{0.14ex}{$\scriptscriptstyle\circ$}\kern-0.8ex%
${\scriptstyle/}$}}%
{\mbox{\raisebox{0.14ex}{$\scriptscriptstyle\circ$}\kern-0.8ex%
${\scriptstyle/}$}}}%
\else\ifx\@currsize\large\mbox{\raisebox{0.2ex}{$\scriptstyle\circ$}\kern-1ex$/$}
\else\ifx\@currsize\small\mbox{\raisebox{0.2ex}{$\scriptstyle\circ$}\kern-1ex$/$}
\else\mbox{\raisebox{0.2ex}{$\scriptstyle\circ$}\kern-1ex$/$} \fi\fi\fi}

\def\plslashblack{\ifx\@currsize\normalsize
{\mathchoice
{\,\mbox{\raisebox{0.2ex}{$\scriptstyle\bullet$}\kern-1ex$\setminus$}}
{\,\mbox{\raisebox{0.2ex}{$\scriptstyle\bullet$}\kern-1ex$\setminus$}}%
{\,\mbox{\raisebox{0.14ex}{$\scriptscriptstyle\bullet$}\kern-0.8ex%
${\scriptstyle\setminus}$}}%
{\,\mbox{\raisebox{0.14ex}{$\scriptscriptstyle\bullet$}\kern-0.8ex%
${\scriptstyle\setminus}$}}}%
\else\ifx\@currsize\large\,\mbox{\raisebox{0.2ex}{$\scriptstyle\bullet$}\kern-1ex$\setminus$}
\else\ifx\@currsize\small\,\mbox{\raisebox{0.2ex}{$\scriptstyle\bullet$}\kern-1ex$\setminus$}
\else\,\mbox{\raisebox{0.2ex}{$\scriptstyle\bullet$}\kern-0.1ex$\setminus$}
\fi\fi\fi}

\newcommand{\ldivmin}{\protect\plslashblack}
\def\prslash{\ifx\@currsize\normalsize
{\mathchoice
{\mbox{\raisebox{0.2ex}{$\scriptstyle\bullet$}\kern-1ex$/$}}
{\mbox{\raisebox{0.2ex}{$\scriptstyle\bullet$}\kern-1ex$/$}}%
{\mbox{\raisebox{0.14ex}{$\scriptscriptstyle\bullet$}\kern-0.8ex%
${\scriptstyle/}$}}%
{\mbox{\raisebox{0.14ex}{$\scriptscriptstyle\bullet$}\kern-0.8ex%
${\scriptstyle/}$}}}%
\else\ifx\@currsize\large\mbox{\raisebox{0.2ex}{$\scriptstyle\bullet$}\kern-1ex$/$}
\else\ifx\@currsize\small\mbox{\raisebox{0.2ex}{$\scriptstyle\bullet$}\kern-1ex$/$}
\else\mbox{\raisebox{0.2ex}{$\scriptstyle\bullet$}\kern-1ex$/$} \fi\fi\fi}
\catcode`@=12
\newcommand{\maxldiv}{\oldiv}
\maketitle

\if 1\IEEE
\thispagestyle{empty}
\pagestyle{empty}
\fi

\begin{abstract}
P-time event graphs (P-TEGs) are specific timed discrete-event systems, in which the timing of events is constrained by intervals.
An important problem is to check, for all natural numbers $d$, the existence of consistent $d$-periodic trajectories for a given P-TEG.
In graph theory, the Proportional-Inverse-Constant-Non-positive Circuit weight Problem (PIC-NCP) consists in finding all the values of a parameter such that a particular parametric weighted directed graph does not contain circuits with positive weight.
In a related paper, we have proposed a strongly polynomial algorithm that solves the PIC-NCP in lower worst-case complexity compared to other algorithms reported in literature.
In the present paper, we show that the first problem can be formulated as an instance of the second; consequently, we prove that the same algorithm can be used to find $d$-periodic trajectories in P-TEGs. 
Moreover, exploiting the connection between the PIC-NCP and max-plus algebra we prove that, given a P-TEG, 
the existence of a consistent 1-periodic trajectory of a certain period is a necessary and sufficient condition for the existence of a consistent $d$-periodic trajectory of the same period, for any value of $d$.
\end{abstract}

\section{Introduction}

P-time event graphs (P-TEGs) constitute a specific class of timed discrete-event systems, namely event graphs in which the sojourn time of tokens in places is constrained by lower and upper bounds.
Events that do not occur in the specified time window correspond, in the real system modeled by the P-TEG, to failure of meeting time specifications.
P-TEGs have been applied for modeling and analyzing time schedules in manufacturing systems such as electroplating lines, baking processes, cluster tools~\cite{becha2013modelling,spacek1999control,declerck2020critical,4456410}.

Timed event graphs are usually assumed to operate under the earliest firing rule~\cite{baccelli1992synchronization}.
This makes their temporal evolution deterministic.
In contrast, in P-TEGs the time evolution of the marking is non-deterministic, since transitions are not required to fire as soon as they are enabled.
In fact, following the earliest firing rule could lead to the violation of upper bound constraints.
For this reason, it is fundamental to characterize the set of all trajectories that are consistent for a P-TEG (i.e., trajectories for which all the constraints are satisfied), and to verify that it is non-empty; in this case, the P-TEG is said to be \textit{consistent}.
In~\cite{zorzenon2020bounded}, we introduced a stronger property than consistency%
; a P-TEG is \textit{boundedly consistent} if it admits at least one consistent trajectory such that the delay between the $k$\textsuperscript{th} firings of every pair of transitions is bounded for all $k$.
Moreover, 
we proved that a P-TEG is boundedly consistent if and only if there exists at least one 1-periodic trajectory\footnote{
We recall that, given a natural number $d$ and a real number $\lambda$, a P-TEG is said to follow a $d$-periodic trajectory of period $\lambda$ when, for all $k$, the $(k+d)$\textsuperscript{th} firing of each transition in the P-TEG occurs $\lambda\times d$ time instants after its $k$\textsuperscript{th} firing. For a formal definition, see~\Cref{se:periodic}.
} that is consistent for the P-TEG.
Since the problem of checking the existence of 1-periodic trajectories can be formulated as a linear programming problem~\cite{becha2013modelling,declerck2017extremum}, and linear programs are known to be solvable (in the worst case) in weakly polynomial time complexity\footnote{
For the definitions of weakly, strongly and pseudo-polynomial time complexity, the reader is referred to~\cite{computers1979,korte2011combinatorial}.
}~\cite{korte2011combinatorial}, bounded consistency can also be verified in the same complexity; on the contrary, no algorithm that checks standard consistency has been found yet.


In graph theory, the \textit{Proportional-Inverse-Constant-Non-positive Circuit weight Problem} (PIC-NCP) consists in finding all values of a real parameter for which a certain parametric weighted directed graph does not contain circuits with positive weight.
In a related paper~\cite{zorzenon2021nonpositive}, we showed that the PIC-NCP can be solved in strongly polynomial time complexity using a modified version of an algorithm discovered by Levner and Kats~\cite{levner1998parametric}, which achieves complexity $\pazocal{O}(n^6)$ in the worst case (where $n$ is the number of nodes in the parametric graph).
In the same paper, we proposed a new algorithm based on max-plus algebra that solves the PIC-NCP in $\pazocal{O}(n^4)$. 
The first aim of the present paper is to show that the problem of finding all periods of $1$-periodic trajectories that are consistent for a P-TEG can be transformed into an instance of the PIC-NCP. 
As an immediate consequence of this result, we prove that bounded consistency can be checked on a P-TEG with $n$ transitions:
	(i) in strongly polynomial time $\pazocal{O}(n^4)$ if the P-TEG has initially at most $1$ token per place;
	(ii) in pseudo-polynomial time $\pazocal{O}(\bar{n}^4)$ (where $\bar{n}$ is evaluated in~\eqref{eq:PTEGtransformation} on page~\pageref{eq:PTEGtransformation}) if the number of tokens initially residing in the places of the P-TEG exceeds $1$ at least in one place.
Additionally, an advantage of our algorithm compared with approaches based on linear programming is that it provides an explicit closed formula for the smallest and largest periods of all admissible 1-periodic trajectories \cite{zorzenon2021nonpositive}.
As a by-product of this result, exploiting the connection between periodic trajectories and the PIC-NCP, in~\Cref{th:fromdto1} we prove a surprising property concerning more general trajectories:
the set of all periods of $d$-periodic trajectories that are consistent for a P-TEG is the same for all natural numbers $d$.
In other terms, a P-TEG admits a 1-periodic trajectory of period $\lambda$ if and only if it admits a $d$-periodic trajectory of the same period, for any value of $d$.
This fact has remarkable practical consequences; whereas verifying the existence of $d$-periodic trajectories (\textit{e.g.}, using linear programming) would normally require a number of operations increasing in $d$, our result implies that it can be done by simply studying 1-periodic trajectories, with significant savings in terms of computational efforts.

The remainder of the paper is organized as follows.
In \Cref{se:alge}, algebraic preliminaries are recalled.
Precedence graphs and the PIC-NCP are introduced in \Cref{se:graphs}.
\Cref{se:PTEGs} contains the formal definition of P-TEGs, and the connection between the PIC-NCP and the problem of checking the existence of consistent $d$-periodic trajectories for P-TEGs is presented in \Cref{se:periodic}.
Finally, conclusions are given in \Cref{se:conclusions}.

\subsection*{Notation}
The set of positive, respectively non-negative, integers is denoted by $\nat$, respectively $\nato$.
The set of non-negative real numbers is denoted by $\R_{\geq 0}$.
Moreover, $\Rmax \coloneqq \R \cup \{-\infty\}$ and $\Rbar \coloneqq \Rmax \cup \{+\infty\}$.

\section{Algebraic preliminaries}\label{se:alge}

In this section, some preliminary concepts from the theory of idempotent semirings and semifields are summarized.
For more details, we refer to~\cite{baccelli1992synchronization,heidergott2014max,hardouin2018control}.

\subsection{Idempotent semirings and semifields}

A set $\D$ endowed with two binary operations $\oplus$ (addition) and $\otimes$ (multiplication), $(\D,\oplus,\otimes)$, is an \textit{idempotent semiring} (or \textit{dioid}), if the following properties hold: $\oplus$ is associative, commutative, idempotent ($\forall a\in\D$, $a\oplus a =a$), and has a neutral element $\varepsilon$; $\otimes$ is associative, distributes over $\oplus$, has a neutral element $e$, and $\varepsilon$ is absorbing for $\otimes$ ($\forall a\in\D$, $a\otimes \varepsilon = \varepsilon \otimes a = \varepsilon$).
The canonical order relation $\preceq$ on $\D$ is defined by: $\forall a,b\in\D$, $a\preceq b \Leftrightarrow a\oplus b = b$.
From this definition, it follows that $\otimes$ is order preserving, i.e., $\forall a,b,c\in\D$, $a\preceq b \Rightarrow a\otimes c \preceq b \otimes c$ and $c\otimes a \preceq c\otimes b$. 
Furthermore, we write that $a\prec b$ if $a\preceq b$ and $a\neq b$.
A \textit{complete dioid} is a dioid that is closed for infinite sums and such that $\otimes$ distributes over infinite sums.
In a complete dioid, $\top\in\D$ denotes the unique greatest element of $\D$, defined as $\top = \bigoplus_{x\in \D} x$.
Moreover, the Kleene star of an element $a\in\D$, denoted $a^*$, is defined by $a^*=\bigoplus_{i\in\nato}a^i$, where $a^0=e$, $a^{i+1}=a\otimes a^i$, and the greatest lower bound $\wedge$ is defined by $a\wedge b = \bigoplus_{\D_{ab}} x$, where $\D_{ab} = \{ x \in \D \ | \ x \preceq a \mbox{ and }  x \preceq b\}$.
The greatest lower bound $\wedge$ is commutative, associative, idempotent, and satisfies the following property: $\forall a,b\in\D$, $a\preceq b \Leftrightarrow a \wedge b = a$.
\begin{remark}\label{re:oplus_separation}
The following equivalence holds: $\forall a,b,c\in\D$, 
$
    a\succeq b \mbox{ and }a\succeq c \ \Leftrightarrow \ a\succeq b\oplus c.
$
Indeed, ($\Leftarrow$) comes from $b\oplus c\succeq b$, $b\oplus c\succeq c$.
($\Rightarrow$) comes from: $a\succeq b\Leftrightarrow a\oplus b = a$, $a\succeq c\Leftrightarrow a\oplus c = a$; therefore, $a\oplus (b\oplus c) = (a\oplus b ) \oplus c = a \oplus c = a$, which is equivalent to $a \succeq b\oplus c$.
Analogously, it is possible to show that: $
    a\preceq b \mbox{ and }a \preceq c \ \Leftrightarrow \ a\preceq b\wedge c.
$
\end{remark}

We can extend operations $\oplus$ and $\otimes$ to matrices as follows: $\forall A,B\in\D^{m\times n}$, $C\in\D^{n\times p}$,
\begin{equation*}
	(A\oplus B)_{ij} = A_{ij} \oplus B_{ij},\quad  
     (A\otimes C)_{ij} = \bigoplus_{k=1}^n (A_{ik}\otimes C_{kj}).
\end{equation*}
Furthermore, the multiplication between a scalar and a matrix is defined as: $\forall \lambda\in\D$, $A\in\D^{m\times n}$, $(\lambda \otimes A)_{ij} = \lambda\otimes A_{ij}$.
If $(\D,\oplus,\otimes)$ is a complete dioid, then $(\D^{n\times n},\oplus,\otimes)$, where $\oplus$ and $\otimes$ are as defined above, is a complete dioid.
In this case, operation $\wedge$ is defined, $\forall A,B\in\D^{n\times n}$, as $(A\wedge B)_{ij} = A_{ij} \wedge B_{ij}$.
The neutral elements for $\oplus$ and $\otimes$ in $(\D^{n\times n},\oplus,\otimes)$ are, respectively, matrices $\pazocal{E}$ and $E_\otimes$, where $\pazocal{E}_{ij} = \varepsilon\ \forall i,j$ and $E_{\otimes ij} = e$ if $i=j$, $E_{\otimes ij} = \varepsilon$ if $i\neq j$.
Furthermore, $A\preceq B\ \Leftrightarrow \ A\oplus B = B\  \Leftrightarrow $ $\forall i,j\ A_{ij} \preceq B_{ij}$.

\begin{definition}[From \cite{brunsch2012duality}]
	A binary operation $\odot$ is a \textit{dual multiplication} in the complete dioid $(\D,\oplus,\otimes)$ if it is associative, distributes over $\wedge$, $e$ is its neutral element, and $\top$ is absorbing for $\odot$.
\end{definition}
Operation $\odot$ can be extended to matrices as: $\forall A\in\D^{m\times n}$, $C\in\D^{n\times p}$, $\lambda\in\D$,
\begin{align*}
	(A\odot C)_{ij} = \bigwedge_{k=1}^n (A_{ik}\odot C_{kj}), \quad
	(\lambda \odot A)_{ij} = \lambda \odot A_{ij}\ ; 
\end{align*}
moreover, if $\odot$ is a dual product in $(\D,\oplus,\otimes)$, then its extension to matrices is a dual product in $(\D^{n\times n},\oplus,\otimes)$.

We say that a complete dioid $(\D,\oplus,\otimes)$ is a \textit{complete idempotent semifield} if, $\forall a\in\D\setminus\{\varepsilon,\top\}$, there exists the multiplicative inverse of $a$, indicated by $a^{-1}$, i.e., $a\otimes a^{-1} = a^{-1} \otimes a = e$.
For notational convenience we write $\varepsilon^{-1}\coloneqq\top$, and $\top^{-1}\coloneqq\varepsilon$.
Note however that $\top\otimes\varepsilon=\varepsilon\otimes\top=\varepsilon$.

\begin{proposition}\label{pr:dual_prod_def}
Let $(\D,\oplus,\otimes)$ be a complete idempotent semifield. 
Then operation $\odot$, defined as
\begin{equation}\label{eq:dual_prod}
	a\odot b = 
	\begin{cases}
        a\otimes b & \mbox{if } a,b\in\D\setminus\{\top\}\\
        \top & \mbox{if } a = \top \mbox{ or } b = \top
	\end{cases}~~~,
\end{equation}
is a dual product for $(\D,\oplus,\otimes)$.
\end{proposition}
\begin{proof}
The only non-trivial property to prove is distributivity of $\odot$ over $\wedge$, i.e., $a\odot(b\wedge c) = (a\odot b)\wedge (a\odot c)$ and $(b\wedge c)\odot a = (b\odot a)\wedge (c\odot a)$, when $a,b,c\in\D\setminus\{\varepsilon,\top\}$; note that this is a direct consequence of Lemma 4.36 of~\cite{baccelli1992synchronization}.
\end{proof}

In the following two propositions (the first of which is an alternative version of Lemma 9-1 from~\cite{cuninghame2012minimax}), we will let $(\D,\oplus,\otimes)$ be a complete idempotent semifield, and $\odot$ be defined as in~\eqref{eq:dual_prod}.

\begin{proposition}\label{pr:inversion_linear_inequality}
For all $ A\in(\D\setminus \{\top\})^{n\times n}$, $x\in(\D\setminus \{\top\})^n$,
\begin{equation}\label{eq:otimes_odot}
	A\otimes x \preceq x \quad \Leftrightarrow \quad x \preceq A^\sharp \odot x, 
\end{equation}
where $(A^\sharp)_{ij} \coloneqq (A_{ji})^{-1}$.
\end{proposition}

\begin{proof}
Let us first show that the proposition holds in the scalar case ($n = 1$), in which $A^\sharp = A^{-1}$.
If $A\in\D \setminus \{\varepsilon, \top\}$, then, from the facts that $\otimes$ is order preserving and $A$ has a multiplicative inverse, $A\otimes x \preceq x \ \Rightarrow\  A^{-1} \otimes A \otimes x = x \preceq A^{-1} \otimes x = A^{-1} \odot x$,
and
$
	x\preceq A^{-1} \odot x = A^{-1}\otimes x\ \Rightarrow\ A\otimes x\preceq A \otimes A^{-1} \otimes x = x.
$
In case $A=\varepsilon$, the proof comes from the absorbing properties of $\varepsilon$ over $\otimes$ and of $\top$ over $\odot$.

In the matrix case, the left-hand side inequality of~\eqref{eq:otimes_odot} can be written as
$\forall i\in\{1,\ldots,n\}$, $\bigoplus_{j=1}^n A_{ij} \otimes x_j \preceq x_i$, which, from \Cref{re:oplus_separation}, is equivalent to, $\forall i,j\in\{1,\ldots,n\}$, $A_{ij} \otimes x_j \preceq x_i$.
From the scalar case, the latter expression is equivalent to, $\forall i,j\in\{1,\ldots,n\}$, $x_j \preceq (A_{ij})^{-1} \odot x_i$, which, using again \Cref{re:oplus_separation}, can be written as $\forall j\in\{1,\ldots,n\}$ $x_j \preceq \bigwedge_{i=1}^{n} (A_{ij})^{-1} \odot x_i $, or, compactly, as $x\preceq A^\sharp \odot x$.
\end{proof}

\begin{proposition} \label{pr:AxleqxleqBx}
For all $A\in (\D\setminus \{\top\})^{n\times n}$, $B\in (\D\setminus \{\varepsilon\})^{n\times n}$ and $x \in (\D\setminus \{\top\})^n$, $A \otimes x \preceq x \preceq B \odot x$ if and only if $(A \oplus B^\sharp) \otimes x\preceq x$.
\end{proposition}
\begin{proof}
The inequalities to the left of the "if and only if" are equivalent to $A\otimes x\preceq x$ and $x\preceq B \odot x$.
From \Cref{pr:inversion_linear_inequality}, the second inequality can be written as $B^\sharp \otimes x\preceq x$.
Using the first statement of \Cref{re:oplus_separation}, the two inequalities can be combined into: $A \maxtimes x \maxplus B^\sharp \maxtimes x = (A \maxplus B^\sharp) \maxtimes x\preceq x$.
\end{proof}

In the following, we recall the definition of tensor (or Kronecker) product and one of its properties, which holds in \textit{commutative dioids}, i.e., dioids in which the product $\otimes$ is commutative.
Let $\MA\in\D^{m\times n}$ and $\MB\in\D^{p\times q}$; the tensor product $\MA\totimes\MB$ is defined as the block matrix
\begin{equation*}\small
	\MA \totimes \MB = 
	\begin{bmatrix}
		\MA_{11}\otimes \MB & \cdots & \MA_{1n} \otimes \MB\\
		\vdots & & \vdots\\
		\MA_{m1} \otimes \MB & \cdots & \MA_{mn} \otimes \MB 
	\end{bmatrix}.
\end{equation*}

\begin{proposition}[Mixed product property] \label{pr:mixed}
	Let $(\D,\oplus,\otimes)$ be a commutative dioid, $\MA\in\D^{m\times n}$, $\MB\in\D^{p\times q}$, $\MC\in\D^{n\times k}$, $D \in\D^{q\times r}$. Then $(\MA\totimes\MB)\otimes(\MC\totimes D) = (\MA\otimes \MC)\totimes (\MB\otimes D)$.
\end{proposition}

\begin{proof}
	The proof is analogous (after replacing $+$ with $\oplus$ and $\times$ with $\otimes$) to the one of Lemma 4.2.10 of~\cite{horn1991topics}.
\end{proof}

\subsection{Max-plus algebra}

The max-plus algebra, $(\Rbar,\oplus,\otimes)$, is the complete idempotent semifield defined as the set $\Rbar$ endowed with operations $\max$, indicated by $\oplus$, and $+$, indicated by $\otimes$.
In $(\Rbar,\oplus,\otimes)$, $\varepsilon=-\infty$, $e=0$, $\top=\infty$, $\wedge$ coincides with the standard $\min$, and $\preceq$ coincides with the standard $\leq$.
Following \Cref{pr:dual_prod_def}, the dual product $\odot$ is defined such that it coincides with $\otimes$, except when one of its arguments is $\top$.
Moreover, note that multiplication $\otimes$ and dual multiplication $\odot$ are commutative in $(\Rbar,\oplus,\otimes)$.
Since $(\Rbar,\oplus,\otimes)$ is a complete idempotent semifield, $(\Rbar^{n\times n},\oplus,\otimes)$ is a complete dioid.
If $A\in\Rbar^{n\times n}$ then $A^\sharp$ is known as Butkovic conjugate matrix~\cite{butkovivc2010max} and coincides, in standard algebra, with $-A^\intercal$.
We will indicate the product between a scalar $\lambda\in\R$ and a matrix $A\in\Rbar^{m\times n}$, $\lambda\otimes A = \lambda \odot A$, simply by $\lambda A$.

\section{Precedence graphs}\label{se:graphs}

A \textit{directed graph} is a pair $(\nodes,\arcs)$ where $\nodes$ is the set of nodes and $\arcs\subseteq \nodes \times \nodes$ is the set of arcs.
Given two nodes $i,j\in\nodes$ of a directed graph, we say that $j$ is a downstream node of $i$ (and $i$ is an upstream node of $j$) if $(i,j)\in\arcs$.
A \textit{weighted directed graph} is a 3-tuple $(\nodes,\arcs,w)$, where $(\nodes,\arcs)$ is a directed graph, and $w:\arcs\rightarrow \R$ is a function that associates a weight $w((j,i))$ to every arc $(j,i)$ of the graph. 

\begin{definition}[Precedence graph]
	Let $\MA\in\Rmax^{n\times n}$.
	The \textit{precedence graph} associated with $\MA$ is the weighted directed graph $\graph( \MA )=(\nodes,\arcs,w)$, where $\nodes=\{1,\ldots,n\}$, $\arcs$ is defined such that there is an arc $(j,i)\in \arcs$ from node $j$ to node $i$ if and only if $ \MA _{ij}\neq -\infty$, and $w$ is such that $w((j,i))= \MA _{ij}$, for every arc $(j,i)\in\arcs$.
	If matrix $ \MA $ depends on some real parameters, $ \MA  = \MA(\lambda_1,\ldots,\lambda_p)$, $\lambda_1,\ldots,\lambda_p\in \R$, we say that $\graph( \MA )$ is a \textit{parametric precedence graph}.
\end{definition}
A \textit{path} $\rho$ in $\graph( \MA )=(\nodes,\arcs,w)$ is a sequence of nodes $(\rho_1,\rho_2,\ldots,\rho_{r+1})$, $r \geq 1$, with arcs from node $\rho_i$ to node $\rho_{i+1}$, such that $\forall i\in\{1,\ldots,r\}$,  $(\rho_i,\rho_{i+1})\in \arcs$.
The length of a path $\rho$ is denoted by $|\rho|_L=r$, and its weight, denoted by $|\rho|_W$, is the max-plus product (standard sum) of the weights of its arcs: $|\rho|_W= \bigotimes_{i=1}^{r}  \MA _{\rho_{i+1},\rho_i}$.
A path $\rho$ is called \textit{circuit} if its initial and final nodes coincide, i.e., if $\rho_1=\rho_{|\rho|_L+1}$.
We recall that, if $\MA\in\Rmax^{n\times n}$, $(\MA^r)_{ij}$ is equal to the maximum weight of all paths of $\graph(\MA)$ from node $j$ to node $i$ of length $r$.
We indicate by $\nonegset$ the set of all precedence graphs that do not contain circuits with positive weight.
The following statement connects max-plus linear inequalities with graphs.
    
\begin{proposition}[From~\cite{baccelli1992synchronization,gallai1958maximum,butkovivc2010max}] \label{pr:v_leq_Bv} 
Let $\MA\in\Rmax^{n\times n}$.
Inequality $\MA\otimes x \preceq x$ has at least one solution $x\in\R^n$ if and only if $\graph(\MA)\in\nonegset$.
Moreover, if $\graph(\MA)\in\nonegset$, then $\MA^*$ can be computed in $\pazocal{O}(n^3)$, and $\{x\in\R^n\ | \ \MA \otimes x \preceq x \} \ =\  \{ \MA^*\otimes u \ | \ u\in\R^n\}$.
\end{proposition}

\subsection{The non-positive circuit weight problem}

Given a parametric precedence graph $\graph( \MA )$, $\MA = \MA(\lambda_1,\ldots,\lambda_p)\in\Rmax^{n\times n}$, the \textit{Non-positive Circuit weight Problem} (NCP) consists in finding all $\lambda_1,\ldots,\lambda_p\in\R$ such that $\graph( \MA )\in\nonegset$.
The solution set of the NCP is indicated by $\solNCP{\MA} \coloneqq \{(\lambda_1,\ldots,\lambda_p)\in\R^p\ | \ \graph(\MA(\lambda_1,\ldots,\lambda_p))\in\nonegset\}$.
Note that, from \Cref{pr:v_leq_Bv}, $\solNCP{\MA}$ coincides with the set of $\lambda_1,\ldots,\lambda_p\in\R$ such that the following inequality admits at least one solution $x\in\R^n$:
\begin{equation}\label{eq:NCP}
	\MA(\lambda_1,\ldots,\lambda_p)\otimes x \preceq x.
\end{equation}

In the following section, we will be particularly interested in a subclass of the NCP, which is described as follows.
Let $\MP,\MI,\MC$ denote three arbitrary $n\times n$ matrices in $\Rmax$ and let us consider the parametric precedence graph $\graphPIC$, in which the weights of the arcs depend on a single parameter $\lambda\in\R$ in a proportional ($\lambda\MP$), inverse ($\lambda^{-1}\MI$) and constant ($\MC$) way, in the max-plus sense.
The NCP on graph $\graphPIC$ is then called \textit{Proportional-Inverse-Constant}-NCP (PIC-NCP).
In standard algebra, the weight of a generic arc $(j,i)$ of $\graphPIC$ has the form $w((j,i)) = \max(\MP_{ij} + \lambda, \MI_{ij} - \lambda, \MC_{ij})$.

\section{P-time event graphs}\label{se:PTEGs}

In this section, P-time event graphs (P-TEGs) are presented following~\cite{zorzenon2020bounded}.%

\subsection{General description of P-time event graphs} \label{su:P-TEGgeneral}

\begin{definition}
An \textit{event graph} is a 4-tuple $(\pazocal{P},\pazocal{T},E,m)$, where $(\pazocal{P}\cup\pazocal{T},E)$ is a directed graph in which the set of nodes is partitioned into the set of places, $\pazocal{P}$ (represented as circles), and the set of transitions, $\pazocal{T}$ (represented as bars), the set of arcs $E$ is such that $E\subseteq (\pazocal{P}\times \pazocal{T})\cup(\pazocal{T}\times \pazocal{P})$, every place has exactly one upstream transition and one downstream transition\footnote{The notion of upstream and downstream places/transitions derives from the definition of upstream and downstream nodes in a directed graph.}, and $m:\pazocal{P}\rightarrow \mathbb{N}_0$ is a map such that $m(p)$ is the initial marking of place $p\in \pazocal{P}$, i.e., the initial number of tokens (represented as dots) in $p$.
\end{definition}

In an event graph, a transition $t\in\pazocal{T}$ is enabled if all its upstream places contain at least $1$ token, or if there are no upstream places of $t$.
When a transition is enabled, it can fire: its firing causes the marking of all its upstream places to decrease by $1$, and the marking of all its downstream places to increase by $1$.

\begin{definition}[From~\cite{CALVEZ19971487}]
A \textit{P-time event graph} (P-TEG) is a 5-tuple $(\pazocal{P},\pazocal{T}$, $E,m,\iota)$, where $(\pazocal{P},\pazocal{T},E,m)$ is an event graph, and $\iota:\pazocal{P}\rightarrow \{ [\tau^-,\tau^+]| \ \tau^-\in \R_{\geq 0},\ \tau^+\in (\R_{\geq 0}\cup \{+\infty\}),\ \tau^-\leq \tau^+\}$
is a map that associates to every place $p\in \pazocal{P}$ a time interval $\iota(p) = [\tau^-_p,\tau^+_p]$, where $\tau_p^-$,  respectively $\tau_p^+$, represent minimal, respectively maximal, sojourn times of tokens in place $p$.
\end{definition}

In addition to the rules stated for event graphs, the dynamics of P-TEGs must respect the following firing rule.
After a token enters a place $p$ at time $\tau$, it must remain there at least until time $\tau + \tau^-_p$ and at most until time $\tau + \tau^+_p$.
If a token does not leave place $p$ before or at time $\tau + \tau^+_p$, then it is said to be dead.
This situation can occur if, for instance, the downstream transition of $p$ is not enabled in time due to a late arrival of tokens in other upstream places of the transition.
We assume that the firing of a transition causes tokens from the upstream places to move instantly to the downstream places; therefore, no delay is associated with the firing of transitions.
Note that, since firing of transitions can occur at any time instant that satisfies the bounds associated with every upstream place, the time evolution of the marking of a P-TEG is non-deterministic.

We recall from~\cite{vspavcek2017analysis} that every P-TEG $(\pazocal{P},\pazocal{T},\arcs,m,\iota)$ in which there exist places with initial marking greater than $1$ can be transformed, without modifying the (non-deterministic) behavior of the original P-TEG, into another one, in which the initial marking of all places is $0$ or $1$.
If $n=|\pazocal{T}|$ is the number of transitions in the original P-TEG, then the number of transitions in the transformed P-TEG is
\begin{equation}\label{eq:PTEGtransformation}
	\bar{n} = n + \sum_{p\in\pazocal{P}} \max(0,m(p)-1)~~.
\end{equation}
Hence, without loss of generality, in the remainder of the paper, only P-TEGs in which every place has zero or unitary initial marking are considered, i.e., $\forall p\in\pazocal{P}$, $m(p)\in\{0,1\}$.

\subsection{Dynamics of a P-TEG in the max-plus algebra}\label{su:P_TEG_dynamics}

In the following, we characterize the time evolution of the marking of a P-TEG as a dynamical system of inequalities in the max-plus algebra.
Let $(\pazocal{P},\pazocal{T},E,m,\iota)$ be a P-TEG with $n$ transitions $t_1,\ldots,t_n\in \pazocal{T}$.
Let the four square matrices $A^0,A^1\in(\R_{\geq 0}\cup\{-\infty\})^{n\times n}$, $B^0, B^1\in (\R_{\geq 0}\cup\{+\infty\})^{n\times n}$ be defined as follows:
if there exists a place $p$ with $m(p)\in\{0,1\}$, upstream transition $t_j$ and downstream transition $t_i$, then $A^{m(p)}_{ij}$ and $B^{m(p)}_{ij}$ represent, respectively, the lower and upper bound of interval $\iota(p)$, i.e., $\iota(p)=[A^{m(p)}_{ij},B^{m(p)}_{ij}]$;
otherwise, if there is no place $p$ with initial marking $\mu\in\{0,1\}$, upstream transition $t_j$ and downstream transition $t_i$, then $A^\mu_{ij}=-\infty$, $B^\mu_{ij}=+\infty$.
Let $x:\mathbb{N}_0\rightarrow \R^n$ be a dater function, i.e., $x_i(k)$ represents the time at which transition $t_i$ fires for the $(k+1)$\textsuperscript{st} time.
It is natural to require that, for every transition, the $(k+1)$\textsuperscript{st} firing cannot occur before the $k$\textsuperscript{th} firing: $\forall i\in\{1,\ldots,n\}$, $\forall k\in \mathbb{N}$ $x_i(k)\geq x_i(k-1)$.
Hence, we will consider only daters that are non-decreasing in $k$.

In order to respect the constraints given by the time bounds $\iota(p)\ \forall p\in \pazocal{P}$, the dater function must satisfy
\begin{equation}\label{eq:dynamics}
\forall k\in\nato\qquad 
\left\{
\begin{array}{rcl}
A^0 \maxtimes x(k) \preceq & x(k) & \preceq B^0\mintimes x(k) \\
A^1\maxtimes x(k) \preceq & x(k+1) & \preceq B^1\mintimes x(k) 
\end{array}
\right.
\end{equation}
(see~\cite{paek2020analysis} for more details).
A non-decreasing trajectory $\{ x(k)\}_{k\in\nato}$ that satisfies~\eqref{eq:dynamics} for every $k\in \mathbb{N}_0$ is said to be \textit{consistent}, or \textit{admissible} for the P-TEG.
Note that following a consistent trajectory $\{x(k)\}_{k\in\nato}$ guarantees that no token dies, since all time constraints associated with every place of the P-TEG are satisfied.

\section{Periodic trajectories and the NCP}\label{se:periodic}

In this section, we show that the problem of finding all $\lambda\in\R_{\geq 0}$ such that a P-TEG admits 1-periodic (and, more generally, $d$-periodic) trajectories of period $\lambda$ is equivalent to the PIC-NCP for a particular choice of matrices $\MP,\ \MI,\ \MC$.

Given a number $d\in\nat$, we say that a trajectory $\{x(k)\}_{k\in\nato}$ is $d$-periodic with period $\lambda\in\R_{\geq 0}$ if, in the max-plus algebra, $\forall k\in\nato$, $x(k+d) = \lambda^d x(k)$.
In standard algebra, $d$-periodic trajectories satisfy: $\forall k\in\nato$, $\forall i\in\{1,\ldots,n\}$, $x_i(k+d) =d\times \lambda + x_i(k)$.
We indicate by $\solPTEG$ the set of all $\lambda\in\R_{\geq 0}$ for which the P-TEG characterized by matrices $A^0,A^1,B^0,B^1$ admits a consistent $d$-periodic trajectory of period $\lambda$, for some $x(0),x(1),\ldots,x(d-1)\in\R^n$.
In the remainder of the section, we consider a P-TEG characterized by matrices $A^0$, $A^1$, $B^0$, $B^1$, and we define: $\MP\coloneqq B^{1\sharp}$, $\MI\coloneqq A^1\oplus E_\otimes$, $\MC\coloneqq A^0\oplus B^{0\sharp}$.
In the following theorem, we show the connection between $d$-periodic trajectories and a variant of the NCP.

\begin{theorem}\label{th:main}
The set $\solPTEG$ coincides with 
$\solNCPt$, where $\TP=\TP(\lambda)$, $\TI=\TI(\lambda)$, $\TC$ are $d\times d$ matrices defined as
\begin{gather*}\small
	\TP \coloneqq 
	\begin{bmatrix}
		-\infty & 0 & -\infty &\cdots & -\infty\\
		-\infty & -\infty & 0 & \cdots & -\infty\\
		\vdots & \vdots & \vdots & \ddots & \vdots\\ 
		-\infty & -\infty & -\infty & \cdots & 0 \\
		\lambda^d & -\infty & -\infty & \cdots & -\infty
	\end{bmatrix},~
	\TI \coloneqq 
	\begin{bmatrix}
		-\infty & -\infty & \cdots & -\infty & \lambda^{-d}\\
		0 & -\infty & \cdots & -\infty & -\infty\\
		-\infty & 0 & \cdots & -\infty & -\infty\\
		\vdots & \vdots & \ddots & \vdots & \vdots\\ 
		-\infty & -\infty & \cdots & 0 & -\infty
	\end{bmatrix}
\end{gather*}
(where expressions $\lambda^d$ and $\lambda^{-d}$ are meant in the max-plus sense and correspond, in the standard algebra, respectively to $d\times \lambda$ and $-d\times\lambda$) and $\TC \coloneqq E_\otimes$.
\end{theorem}

\begin{proof}
Let us rewrite~\eqref{eq:dynamics} adding conditions $\forall k\in\nato$ $x(k+d)=\lambda^d x(k)$ ($d$-periodicity), $x(k+1)\succeq x(k)$ (non-decreasingness), 
using \Cref{pr:inversion_linear_inequality,pr:AxleqxleqBx}, and the fact that $\otimes$ is order preserving:
\begin{gather*}\small
\forall k\in\nato \qquad 
\left\{
\begin{array}{rl}
(A^0 \maxplus B^{0\sharp}) \maxtimes x(k) & \preceq x(k)\\
A^1\otimes x(k) &\preceq x(k+1)\\
B^{1\sharp} \otimes x(k+1) & \preceq x(k)\\
E_\otimes\otimes x(k) &\preceq x(k+1)\\
& \vdots \\
(A^0 \maxplus B^{0\sharp}) \maxtimes x(k+d-1) & \preceq x(k+d-1)\\
\lambda^{-d} A^1\otimes x(k+d-1) &\preceq x(k)\\
\lambda^d B^{1\sharp} \otimes x(k) & \preceq x(k+d-1)\\
\lambda^{-d}E_\otimes\otimes x(k+d-1) &\preceq x(k)\\
\end{array}
\right. 
\end{gather*}
Using \Cref{re:oplus_separation}, the distributive property of $\maxtimes$ over $\maxplus$ and the definition of tensor product, the system above can be written as
\[
\forall k\in\nato \quad 
((\TP\totimes\MP)\oplus(\TI\totimes\MI)\oplus(\TC\totimes\MC))\maxtimes \tilde{x}(k) \preceq \tilde{x}(k)~,
\]
where $\tilde{x}(k) \coloneqq [x^\intercal(k),x^\intercal(k+1),\ldots,x^\intercal(k+d-1)]^\intercal\in\R^{dn}$.
Note that, if $\tilde{x}(0)=\tilde{x}_0$ satisfies the above inequality, from the fact that $\lambda^d$ admits a multiplicative inverse $\lambda^{-d}$, then $\forall k\in\nato$, $\tilde{x}(k+1)$ automatically satisfies it.
Therefore, there exists an admissible $d$-periodic trajectory of period $\lambda$ if and only if inequality 
\begin{equation}\label{eq:aux}
((\TP\totimes\MP)\oplus(\TI\totimes\MI)\oplus(\TC\totimes\MC))\maxtimes \tilde{x}_0 \preceq \tilde{x}_0~,
\end{equation}
admits a solution $\tilde{x}_0\in\R^{dn}$.
Hence, from~\eqref{eq:NCP}, finding all $\lambda$s for which there exists a $d$-periodic trajectory of period $\lambda$ that is consistent for the P-TEG is equivalent to solving the NCP for the parametric precedence graph $\graphNCPt$.
\end{proof}

\begin{remark}\label{re:initial}
Given a $\lambda\in\solPTEG$, it is possible to characterize, in time $\pazocal{O}((dn)^3)$, the set of vectors $\tilde{x}_0=[x_0^\intercal,\ldots,x_{d-1}^\intercal]^\intercal$ such that the $d$-periodic trajectory $x(0)=x_0,\ldots, x(d-1)=x_{d-1}$, $\forall k\in\nato$, $x(k+d)=\lambda^dx(k)$ is consistent for the P-TEG.
Indeed, as discussed in the proof of \Cref{th:main}, this set coincides with the set of solutions $\tilde{x}_0\in\R^{dn}$ to~\eqref{eq:aux}, which, from \Cref{pr:v_leq_Bv}, is equivalent to
$
\{((\TP\totimes\MP)\oplus(\TI\totimes\MI)\oplus(\TC\totimes\MC))^*\otimes u \ | \ u\in\R^{dn} \}.
$
\end{remark}

Since $\TP[1]=\lambda$, $\TI[1]=\lambda^{-1}$, $\TC[1]=0$, an immediate consequence of \Cref{th:main} is that $\solPTEG[1]$ coincides with the solution set of the PIC-NCP on graph $\graphPIC$, i.e., $\solNCP{\lambda\MP\oplus\lambda^{-1}\MI\oplus\MC}$.
As the PIC-NCP can be solved in time complexity $\pazocal{O}(n^4)$~\cite{zorzenon2021nonpositive}, this proves that we can compute the set of $\lambda$s for which there exists a 1-periodic trajectory with period $\lambda$ that is consistent for a P-TEG in the same complexity.

On the other hand, since $\TP\totimes\MP$, $\TI\totimes\MI$, $\TC\totimes\MC\in\Rmax^{dn\times dn}$,
it is natural to think that the complexity for computing $\solPTEG$ would grow with $d$.
However, the following result gives a surprisingly simple characterization of the admissible periods of $d$-periodic trajectories.

\begin{theorem}\label{th:fromdto1}
For all $d\in\nat$, the set $\solPTEG$ coincides with $\solPTEG[1]$.
\end{theorem}

\begin{proof}
Since every 1-periodic trajectory is also $d$-periodic, it is obvious that $\solPTEG\supseteq \solPTEG[1]$.
Therefore, we only need to prove that $\solPTEG\subseteq\solPTEG[1]$; from \Cref{th:main}, this is equivalent to prove that $\solNCPt
\subseteq\solNCP{\lambda\MP\oplus\lambda^{-1}\MI\oplus\MC}$.
In order to do so, in the following we suppose that, for a given $\lambda\in\R_{\geq0}$, there exists a circuit in $\graphPIC$ from node $i\in\{1,\ldots,n\}$ of length $h$ and positive weight; we will prove that, as a consequence, there exists a circuit in $\graphNCPt$
from node $i$ of length $h\times d$ and positive weight.
Algebraically, this is the same as proving that, if $((\lambda \MP \oplus \lambda^{-1} \MI \oplus \MC)^h)_{ii}\succ0$, then 
$((((\TP\totimes\MP)\oplus(\TI\totimes\MI)\oplus(\TC\totimes\MC))^{h})^d)_{ii}\succ0$.

By expanding the $h$-th power of the trinomial $\lambda\MP\oplus \lambda^{-1}\MI\oplus \MC$ we observe that, since $\oplus$ is idempotent, there must exist a certain matrix $M$ of the form $M = \lambda^{h_\MP-h_\MI} X_1 \otimes \cdots \otimes X_h$, with $X_1,\ldots,X_h\in\{\MP,\MI,\MC\}$, obtained by multiplying, in some order, $h_\MP$ times matrix $\lambda \MP$, $h_\MI$ times matrix $\lambda^{-1} \MI$ and $h_\MC$ times matrix $\MC$, with $h_\MP+h_\MI+h_\MC=h$, such that $((\lambda \MP \oplus \lambda^{-1} \MI \oplus \MC)^h)_{ii} = M_{ii} \succ0$.

By expanding the $h$-th power of the trinomial $(\TP\totimes \MP) \oplus (\TI\totimes \MI) \oplus (\TC \totimes \MC)$ we get, by idempotency,
\begin{gather*}
	((((\TP\totimes\MP) \oplus (\TI\totimes\MI)\oplus (\TC\totimes\MC))^{h})^d)_{ii} \succeq\\
	\succeq (((\T{1}\totimes X_1)\otimes (\T{2}\totimes X_2)\otimes \cdots \otimes (\T{h}\totimes X_h))^d)_{ii}~.
\end{gather*}
Using \Cref{pr:mixed} $(h\times d)-1$ times (after writing the $d$-th power as the product of $d$ terms), the latter expression can be written as
\begin{gather}\label{eq:aux1}
	((\T{1}\otimes \cdots \otimes \T{h})^d\totimes (X_1\otimes \cdots \otimes X_h)^d)_{ii}~. 
\end{gather}
It is easy to verify that the following properties hold on matrices $\TP$ and $\TI$: $(\TP)^d = \lambda^d E_\otimes$, $(\TI)^d=\lambda^{-d} E_\otimes$, and $\TP\otimes \TI = \TI \otimes \TP = \TC = E_\otimes$.
Therefore,
\begin{align*}
	(\T{1}\otimes \cdots \otimes \T{h})^d
	&=\begin{dcases}
		((\TP)^{h_\MP-h_\MI})^d & \mbox{if } h_\MP\geq h_\MI\\
		((\TI)^{h_\MI-h_\MP})^d & \mbox{otherwise} 
	\end{dcases}\\
	&=(\lambda^{h_\MP-h_\MI})^d E_\otimes~,
\end{align*}
and~\eqref{eq:aux1} simplifies to
\begin{gather*}
	((\lambda^{h_\MP-h_\MI} )^d E_\otimes\totimes (X_1\otimes \cdots \otimes X_h)^d)_{ii} =\\
	= ((\lambda^{h_\MP-h_\MI}X_1\otimes\cdots\otimes X_h)^d)_{ii}= (M^d)_{ii} \succeq M_{ii}\succ0.
\qedhere\end{gather*}
\end{proof}

\subsection{Example}\label{su:example}

In the following, we illustrate the previous theorems by a simple example taken from~\cite{zorzenon2020bounded}.
Let us consider the P-TEG shown in Fig.~\ref{fig:TEG_example}, characterized by matrices 
\begin{equation*}
\if 1\IEEE
\small
\fi
A^0 =     
    \begin{bmatrix}
        -\infty & -\infty & -\infty \\
        2 & -\infty & -\infty \\
        6 & 0.5 & -\infty
    \end{bmatrix},\ 
A^1 = \begin{bmatrix}
	-\infty & 0 & -\infty \\
	-\infty & -\infty & 0.5 \\
	-\infty & -\infty & 0
\end{bmatrix},
\end{equation*}
\begin{equation*}
\if 1\IEEE
\small
\fi
B^0 = \begin{bmatrix}
	+\infty & +\infty & +\infty \\
	3 & +\infty & +\infty \\
	+\infty & +\infty & +\infty
\end{bmatrix},\
B^1 = \begin{bmatrix}
        +\infty & +\infty & +\infty \\
        +\infty & +\infty & +\infty \\
        +\infty & +\infty & 4
    \end{bmatrix}.
\end{equation*}
\begin{figure}[t]
    \centering
    \resizebox{
\if 1\IEEE
    	.7\linewidth
\else
    	.6\linewidth
\fi
}{!}{
    \begin{tikzpicture}[node distance=.1cm and 1.5cm,>=stealth',bend angle=45,thick]

\node [transitionV,label=below:$t_1$] (x1) {};
\node [place,tokens=0,label=below:{$[2,3]$}] (p21) [right= of x1] {};
\node [place,tokens=1,label=above:{$[0,+\infty]$}] (p12) [above= of p21] {};
\node [transitionV,label=below:$t_2$] (x2) [right=of p21] {};
\node (p22) [below= of x2] {};
\node [place,tokens=0,label=below:{$[0.5,+\infty]$}] (p32) [right= of x2] {};
\node [place,tokens=1,label=above:{$[0.5,+\infty]$}] (p23) [above= of p32] {};
\node [transitionV,label=below:$t_3$] (x3) [right=of p32] {};
\node [place,tokens=1,label=below:{$[0,4]$}] (p33) [right= of x3] {};
\node [place,tokens=0,label=below:{$[6,+\infty]$}] (p31) [below= of p22] {};

\draw (x1) edge[->] (p21);
\draw (p21) edge[->] (x2);
\draw (x2) edge[->] (p32);
\draw (p32) edge[->] (x3);
\draw (x1) edge[bend right=25,->] (p31);
\draw (p31) edge[bend right=25,->] (x3);
\draw (x2) edge[bend right=20,->] (p12);
\draw (p12) edge[bend right=20,->] (x1);
\draw (x3) edge[bend right=20,->] (p23);
\draw (p23) edge[bend right=20,->] (x2);
\draw (x3) edge[bend right=20,->] (p33);
\draw (p33) edge[bend right=20,->] (x3);
\end{tikzpicture}
    }
    \caption{Example of P-TEG.}
    \label{fig:TEG_example}
\end{figure}
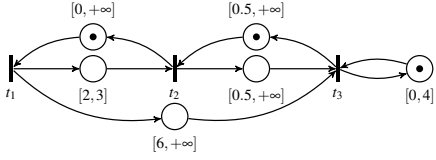
Let $\MP=B^{1\sharp}$, $\MI=A^1\oplus E_\otimes$, $\MC=A^0\oplus B^{0\sharp}$.
From \Cref{th:main,th:fromdto1}, for all $d\in\nat$, given a number $\lambda\in\R$ there exists a consistent $d$-periodic trajectory of period $\lambda$ if and only if the precedence graph $\graphPIC$ (shown in Fig.~\ref{fig:TEG_example_prec_graph}) does not contain circuits with positive weight.
Using the algorithm presented in~\cite{zorzenon2021nonpositive}, we obtain:
$
\solPTEG=\solNCP{\lambda\MP\oplus\lambda^{-1}\MI\oplus\MC} = \{\lambda \ | \ 3.5\leq \lambda \leq 4\}.
$
For instance, let us consider $\lambda=4$ and $d=2$; from \Cref{re:initial}, the 2-periodic trajectory $x(0)=x_0=[0,2.5,6]^\intercal$, $x(1)=x_1=[3.5,6.5,10]^\intercal$, $\forall k\in\nato$, $x(k+2)=4^2x(k)=8x(k)$ is consistent for the P-TEG, since, for $u = [0,\ 0,\ 0,\ 0,\ 0,\ 0]^\intercal$, we have
\begin{gather*}\small
	\TP[2] = 
\begin{bmatrix}
	-\infty & 0 \\
	8 & -\infty
\end{bmatrix},~
	\TI[2] = 
\begin{bmatrix}
	-\infty & -8\\
	0 & -\infty
\end{bmatrix},~
	\TC[2] = 
\begin{bmatrix}
	0 & -\infty\\
	-\infty & 0
\end{bmatrix},
\end{gather*}
\begin{gather*}\small
\left(
\left(
\TP[2]
\totimes \MP\right)
\oplus
\left(
\TI[2]
\totimes \MI\right)
\oplus
\left(
\TC[2]
\totimes \MC\right)
\right)^*\otimes u =
\end{gather*}
\begin{gather*}
\if 1\IEEE
\small
\fi
=\begin{bmatrix}
         0   &-3   &-6.5   &-4.5   &-7.5  &-10.5\\
    2.5         &0   &-3.5   &-1.5   &-4.5   &-7.5\\
    6    &3         &0    &2   &-1   &-4\\
    3.5    &0.5   &-2.5         &0   &-3   &-6.5\\
    6.5    &3.5    &0.5    &2.5         &0   &-3.5\\
   10    &7    &4    &6    &3         &0
\end{bmatrix}
\otimes 
\begin{bmatrix}
	0\\0\\0\\0\\0\\0
\end{bmatrix}=
\begin{bmatrix}
	x_0\\
	x_1\\
\end{bmatrix}.
\end{gather*}

\section{Conclusions}\label{se:conclusions}

The present paper shows the connection between two problems: the PIC-NCP, and the problem of finding, given a natural number $d$, the set of all real numbers $\lambda$ such that there exists a $d$-periodic trajectory of period $\lambda$ that is consistent for a P-TEG.
Indeed, based on their description in the max-plus algebra, we formulate the second problem as a particular instance of the first one.
Moreover, we prove that the solution set of the second problem is invariant with respect to $d$.
In a related paper~\cite{zorzenon2021nonpositive}, we have suggested an algorithm that solves the PIC-NCP in time complexity $\pazocal{O}(n^4)$, where $n$ is the number of nodes in the graph; this proves that the existence of consistent $d$-periodic trajectories can be checked in strongly polynomial time $\pazocal{O}(n^4)$, where $n$ is the number of transitions of a P-TEG with initially $0$ or $1$ token per place.
Since every P-TEG can be transformed into one with at most $1$ initial token per place, a consequence is that the same property can be checked in generic P-TEGs in time complexity $\pazocal{O}(\bar{n}^4)$, where $\bar{n}$ is defined in~\eqref{eq:PTEGtransformation}.
In future work, we aim at assessing, through numerical simulations, the advantages of the algorithm proposed in~\cite{zorzenon2021nonpositive} in terms of execution time, and at extending its applicability to safe P-time Petri nets.

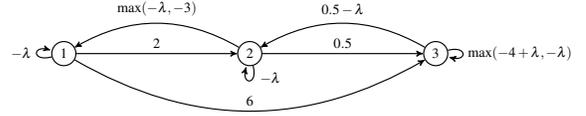
\begin{figure}[t]
	\centering
	\resizebox{
\if 1\IEEE
.9\linewidth
\else
.8\linewidth
\fi
}{!}{
	\begin{tikzpicture}[node distance=1cm and 4cm,>=stealth',bend angle=45,thick]

\node [place,tokens=0,label=center:{$1$},] (P1) {};
\node [place,tokens=0,label=center:{$2$},right= of P1] (P2) {};
\node [place,tokens=0,label=center:{$3$},right= of P2] (P3) {};

\draw (P1) edge[->] node[auto,swap,label=above:{$2$}] {} (P2);
 \draw (P1) edge[loop left,->] node {$-\lambda$} (P1);
 \draw (P2) edge[loop below,->] node [label={[yshift=.2cm]right:{$-\lambda$}}] {} (P2);
\draw (P3) edge[loop right,->] node {$\max(-4+\lambda,-\lambda)$} (P3);
\draw (P2) edge[bend right=30,->] node[auto,swap,label={[yshift=-.2cm]above:{$\max(-\lambda,-3)$}}] {} (P1);
\draw (P2) edge[->] node[auto,swap,label=above:{$0.5$}] {} (P3);
\draw (P1) edge[bend right=30,->] node[auto,swap,label=above:{$6$}] {} (P3);
\draw (P3) edge[bend right=30,->] node[auto,swap,label={[yshift=-.2cm]above:{$0.5-\lambda$}}] {} (P2);

\end{tikzpicture}
	}
	\caption{Parametric precedence graph associated with the P-TEG of Fig.~\ref{fig:TEG_example}.}
	\label{fig:TEG_example_prec_graph}
\end{figure}

\if 1\IEEE
\bibliographystyle{plain}
\else
\bibliographystyle{plainnat}
\fi
\bibliography{references}
\end{document}